\documentclass[11pt]{article}

\usepackage[margin=1in]{geometry}
\usepackage{color}
\usepackage{amssymb}
\usepackage{amsmath}
\usepackage{amsthm}
\usepackage{url}
\usepackage{framed}
\usepackage{enumitem}
\usepackage{amsfonts}
\usepackage{booktabs}
\usepackage{siunitx}

\usepackage{graphicx}

\newtheorem{theorem}{Theorem}
\newtheorem{definition}[theorem]{Definition}

\newtheorem{lemma}[theorem]{Lemma}
\newtheorem{corollary}[theorem]{Corollary}

\newcommand\T{\rule{0pt}{2.6ex}}
\newcommand\B{\rule[-1.2ex]{0pt}{0pt}}

\newcommand{\C}{\mathcal C}
\newcommand{\RR}{\mathbb R}

\begin{document}

\title{On the tightness of an SDP relaxation of $k$-means}

\author{Takayuki Iguchi, Dustin G. Mixon, Jesse Peterson, Soledad Villar}
\date{}
\maketitle

\begin{abstract}
Recently, \cite{relax} introduced an SDP relaxation of the $k$-means problem in $\RR^m$. In this work, we consider a random model for the data points in which $k$ balls of unit radius are deterministically distributed throughout $\RR^m$, and then in each ball, $n$ points are drawn according to a common rotationally invariant probability distribution.
For any fixed ball configuration and probability distribution, we prove that the SDP relaxation of the $k$-means problem exactly recovers these planted clusters with probability $1-e^{-\Omega(n)}$ provided the distance between any two of the ball centers is $>2+\epsilon$, where $\epsilon$ is an explicit function of the configuration of the ball centers, and can be arbitrarily small when $m$ is large.
\end{abstract}

\section{Introduction}

Clustering is one central task in unsupervised machine learning.
The problem consists of partitioning a given finite set $P$ into $k$ subsets $\C=\{\C_1,\ldots,\C_k\}$ such that some dissimilarity function is minimized. Usually the similarity criterion is chosen ad hoc with an application in mind. A particularly common clustering criterion is the $k$-means objective. Let $P \subset \RR^m$ a finite set. For $\C_i\subset P$ let $c_i$ be its centroid $c_i=\frac{1}{|\C_i|}\sum_{x\in \C_i} x$. Then the $k$-means problem is
\begin{equation} \label{kmeans}
\mathop{\min_{\C_1\cup\ldots\cup \C_k= P}}_{\C_i\cap \C_j=\emptyset} \sum_{i=1}^k \sum_{x\in \C_i} {\| x-c_i \|^2}.
\end{equation}

Problem \eqref{kmeans} is NP-hard in general~\cite{jms06}. A popular approach to solving this problem is the heuristic algorithm by Lloyd, also known as the $k$-means algorithm~\cite{Lloyd}. This algorithm alternates between calculating centroids of proto-clusters and reassigning points according to the nearest centroid.
Lloyd's algorithm (and its variants~\cite{Arthur07, Lloyd06}) may, in general, converge to local minima of the $k$-means objective (see for example section 5 of \cite{relax}). Furthermore, the output of Lloyd's algorithm does not indicate how far it is from optimal. 
As such, a slower algorithm that emits such a certificate may be preferable.

Along these lines, convex relaxations provide a framework to attack NP-hard combinatorial problems. This framework is known as the ``relax and round'' paradigm. Given an optimization problem, first relax the feasibility region to a convex set, optimize subject to this larger set, and then round this optimal solution to a point in the original feasibility region. One may seek approximation guarantees in this framework by relating the value of the rounded solution to the value of the optimal solution.
Convex relaxations of clustering problems have been studied~\cite{peng2005new, peng2007approximating}, and a particular relaxation of $k$-means is known to satisfy an approximation ratio~\cite{kanungo02}.

Sometimes, the rounding step of the approximation algorithm is unnecessary because the convex relaxation happens to find a solution that is feasible in the original problem. This phenomenon is known as exact recovery, tightness, or integrality of the convex relaxation. Note that when exact recovery occurs, the algorithm not only provides a solution, but also a certificate of its optimality, thanks to convex duality.
This paper focuses on exact recovery under a particular convex relaxation of the $k$-means problem.

\subsection{Integrality of convex relaxations of geometric clustering}

When is a convex relaxation of geometric clustering tight?
This question seems to have first appeared in~\cite{elhamifar2012finding}, where the authors study an LP relaxation of the $k$-median objective (a problem which is similar to $k$-means). That first paper proves tightness of the relaxation provided the set of points $P$ admits a partition into $k$ clusters of equal size, and the separation distance between any two clusters is sufficiently large. 
Later on, \cite{Nellore_Kmedians} studied integrality of another LP relaxation to the $k$-median objective.
This paper introduced a distribution on the input $P$, which we refer to as the stochastic ball model:

\begin{definition}[$(\mathcal D, \gamma, n)$-stochastic ball model]
\label{stochastic_balls}
Let $\{\gamma_a\}_{a=1}^k$ be ball centers in $\RR^m$. For each $a$, draw iid vectors $\{r_{a,i}\}_{i=1}^n$ from some rotation-invariant distribution $\mathcal D$ supported on the unit ball. The points from cluster $a$ are then taken to be $x_{a,i}:= r_{a,i} + \gamma_a$. 
\end{definition}

Table~\ref{table} summarizes the state of the art for recovery guarantees under the stochastic ball model. 
In \cite{Nellore_Kmedians}, it was shown that the LP relaxation of $k$-medians will, with high probability, recover clusters drawn from the stochastic ball model provided the smallest distance between ball centers is $\Delta \geq 3.75$.  Note that exact recovery only makes sense for $\Delta>2$ (i.e., when the balls are disjoint). Once $\Delta \geq 4$, any two points within a particular cluster are closer to each other than any two points from different clusters, and so in this regime, cluster recovery follows from a simple thresholding.

\begin{table}
\begin{center}
\begin{tabular}{llll} \hline
    {\textbf{Method}} & {\textbf{Sufficient Condition}} & {\textbf{Optimal?}} & {\textbf{Reference}} \T\B \\ \hline\hline
    Thresholding  & $\Delta\geq4$ & Yes & (simple exercise) \T\B \\ \hline
    $k$-medians LP  & $\Delta\geq4$  & No & Theorem~2 in \cite{elhamifar2012finding} \T \\
      & $\Delta\geq3.75$  & No & Theorem~1 in \cite{Nellore_Kmedians} \\
      & $\Delta>2$  & Yes &  Theorem~1 in \cite{relax} \B \\ \hline
    $k$-means LP  & $\Delta\geq4$   & Yes & Theorem~9 in \cite{relax} \T \B \\ \hline
    $k$-means SDP  & $\Delta\geq2\sqrt2(1+ 1/\sqrt m)$  & No  & Theorem~3 in \cite{relax} \T\\
      & $\Delta>2$ & Yes*  & Conjecture~4 in \cite{relax}   \\
      & $\Delta\geq2+k^2\operatorname{Cond}(\gamma)/m$  & No*  & Theorem~\ref{main_theorem} \B   \\ \hline
\end{tabular}
\end{center}
\caption{\label{table} Summary of cluster recovery guarantees under the stochastic ball model.
The second column reports sufficient separation between ball centers in order for the corresponding method to provably give exact recovery with high probability.
(*) We report whether these bounds are optimal under the assumption of Conjecture~4 in \cite{relax}.
}
\end{table}

For the $k$-means problem, \cite{relax} provides an SDP relaxation and demonstrates exact recovery in the regime $\Delta>2\sqrt2(1+ 1/\sqrt m)$, where $m$ is the dimension of the Euclidean space. That work also conjectures that the result holds for optimal separation $\Delta> 2$.
The present work demonstrates tightness given near-optimal separation:

\begin{theorem}[Main result]
\label{main_theorem}
The $k$-means SDP relaxation \eqref{eq.kmeansSDP} from \cite{relax} recovers the planted clusters in the $(\mathcal D, \gamma,n)$-stochastic ball model with probability $1-e^{-\Omega_{\mathcal D, \gamma}(n)}$ provided 
$$\Delta >2 + \frac{k^2}{m}\operatorname{Cond}(\gamma),$$
where
$$\operatorname{Cond}(\gamma):= \frac{\max_{a,b\in\{1,\ldots,k\}, a\neq b} \|\gamma_a-\gamma_b\|^2}{\min_{a,b\in\{1,\ldots,k\}, a\neq b} \|\gamma_a-\gamma_b\|^2}$$ 
\end{theorem}

Our proof of Theorem~\ref{main_theorem} follows the strategy of \cite{relax}, namely, to identify a dual certificate of the SDP, and then show that this certificate exists for a suitable regime of $\Delta$'s under the stochastic ball model.
Figure~\ref{phase_transitions} provides numerical simulations that illustrate the empirical performance of our dual certificate in comparison with the one provided in \cite{relax}. 

\begin{figure} 
\includegraphics[height=0.47\textheight]{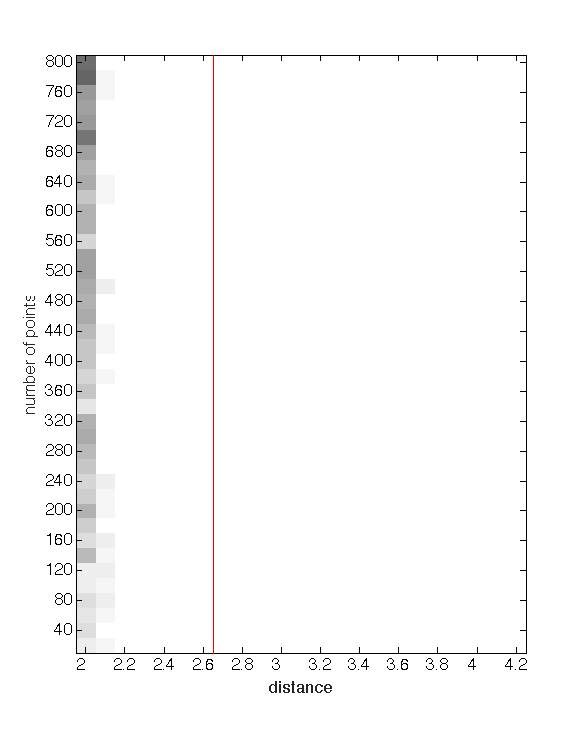}
\includegraphics[height=0.47\textheight]{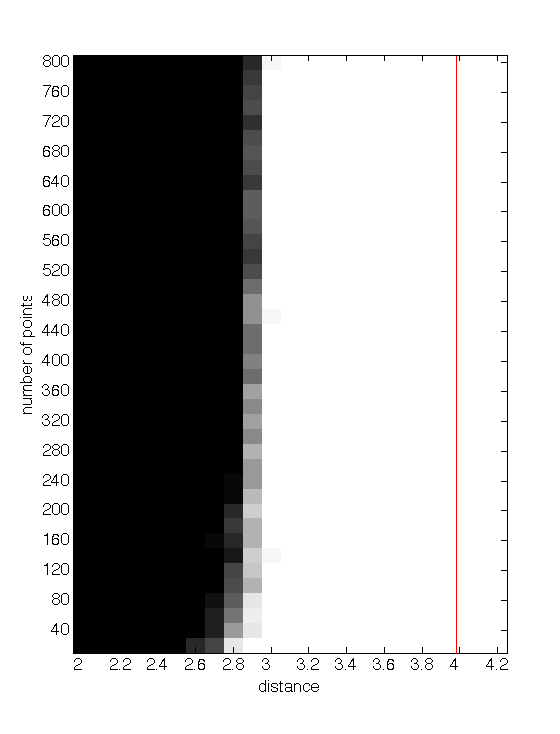}
\caption{\label{phase_transitions} Frequency of successful certification for the SDP relaxation of k-means (left: our certificate, right: certificate from \cite{relax}).  Lighter color represents higher probability of success. Area to the right of the vertical line corresponds to the regime where exact recovery was proven with high probability.
We generate 30 random instances of our $(\mathcal D, \gamma, n)$-stochastic ball model for each given distance and number of points.  Here we take $\mathcal D$ as the uniform distribution in the unit ball in $\mathbb R^6$.
}
\end{figure}

This paper is organized as follows. Section \ref{sec:background} formulates a semidefinite relaxation of $k$-means and derives its dual. Section \ref{sec:deterministic} provides deterministic conditions that guarantee the solution of the relaxation provided is feasible for $k$-means. Section \ref{sec:model} proves Theorem \ref{main_theorem}, showing the deterministic conditions are satisfied with high probability under the $(\mathcal D, \gamma, n)$-stochastic ball model.
\section{Background} \label{sec:background}

Given $P\subseteq\mathbb{R}^m$ with $|P|=N$, we seek to solve the $k$-means problem \eqref{kmeans} which is well-known to be equivalent to
\begin{alignat}{2}
\label{eq.kmeans}
& \text{minimize}  &       & \sum_{t=1}^k\frac{1}{|A_t|}\sum_{x_i,x_j\in A_t}\|x_i-x_j\|_2^2 \\
\nonumber
& \text{subject to}& \quad & A_1\sqcup\cdots\sqcup A_k=P
\end{alignat}
This problem is NP-hard in general~\cite{Aloise}.
However, many instances of this problem can be solved by relaxing to the following SDP:
\begin{alignat}{2}
\label{eq.kmeansSDP}
& \text{maximize}  &       & -\operatorname{Tr}(DX) \\
\nonumber
& \text{subject to}& \quad & 
\begin{aligned}[t]
\operatorname{Tr}(X)&=k\\
X1&=1\\
X&\geq0\\
X&\succeq0
\end{aligned}
\end{alignat}
Here, $D$ denotes the matrix whose $(i,j)$th entry is $\|x_i-x_j\|_2^2$.
Observe that \eqref{eq.kmeansSDP} is indeed a relaxation of \eqref{eq.kmeans}:
Let $1_A$ denote the indicator function of $A\subseteq\{1,\ldots,N\}$.
Then taking $X:=\sum_{t=1}^k\frac{1}{|A_t|}1_{A_t}1_{A_t}^\top$ gives that
\begin{align*}
\operatorname{Tr}(DX)
&=\operatorname{Tr}\bigg(D\sum_{t=1}^k\frac{1}{|A_t|}1_{A_t}1_{A_t}^\top\bigg)
=\sum_{t=1}^k\frac{1}{|A_t|}\operatorname{Tr}(D1_{A_t}1_{A_t}^\top)
=\sum_{t=1}^k\frac{1}{|A_t|}1_{A_t}^\top D1_{A_t}\\
&=\sum_{t=1}^k\frac{1}{|A_t|}\sum_{x_i,x_j\in A_t}\|x_i-x_j\|_2^2.
\end{align*}
Also, $X$ is clearly feasible in \eqref{eq.kmeansSDP}, and so we conclude that the SDP is a relaxation of the $k$-means problem \eqref{eq.kmeans}.

To derive the dual of \eqref{eq.kmeansSDP}, we will leverage the general setting from cone programming~\cite{Mixon:15}, namely, that given closed convex cones $K$ and $L$, the dual of
\begin{alignat}{2}
\label{eq.primal}
& \text{maximize}  &       & \langle c,x\rangle \\
\nonumber
& \text{subject to}& \quad & 
\begin{aligned}[t]
b-Ax&\in L\\
x&\in K
\end{aligned}
\end{alignat}
is given by
\begin{alignat}{2}
\label{eq.dual}
& \text{minimize}  &       & \langle b,y\rangle \\
\nonumber
& \text{subject to}& \quad & 
\begin{aligned}[t]
A^* y-c&\in K^*\\
y&\in L^*
\end{aligned}
\end{alignat}
where $A^*$ denotes the adjoint of $A$, while $K^*$ and $L^*$ denote the dual cones of $K$ and $L$, respectively.
In our case, $c=-D$, $x=X$, and $K$ is simply the cone of positive semidefinite matrices (as is $K^*$).
Before we determine $L$, we need to interpret the remaining constraints in \eqref{eq.kmeansSDP}.
To this end, we note that $\operatorname{Tr}(X)=k$ is equivalent to $\langle X,I\rangle=k$, $X1=1$ is equivalent to having
\[
\bigg\langle X,\frac{1}{2}(e_i1^\top+1e_i^\top)\bigg\rangle
=1
\qquad
\forall i\in\{1,\ldots,N\},
\]
and $X\geq0$ is equivalent to having
\[
\bigg\langle X,\frac{1}{2}(e_ie_j^\top+e_je_i^\top)\bigg\rangle
\geq0
\qquad
\forall i,j\in\{1,\ldots,N\},~i\leq j.
\]
(These last two equivalences exploit the fact that $X$ is symmetric.)
As such, we can express the remaining constraints in \eqref{eq.kmeansSDP} using a linear operator $A$ that sends any matrix $X$ to its inner products with $I$, $\{\frac{1}{2}(e_i1^\top+1e_i^\top)\}_{i=1}^N$, and $\{\frac{1}{2}(e_ie_j^\top+e_je_i^\top)\}_{i,j=1,i\leq j}^N$.
The remaining constraints in \eqref{eq.kmeansSDP} are equivalent to having $b-Ax\in L$, where $b=k\oplus 1\oplus 0$ and $L=0\oplus0\oplus\mathbb{R}_{\geq0}^{N(N+1)/2}$.
Writing $y=z\oplus \alpha\oplus(-\beta)$, the dual of \eqref{eq.kmeansSDP} is then given by
\begin{alignat}{2}
\label{eq.kmeansSDPdual}
& \text{minimize}  &       & kz+\sum_{i=1}^N\alpha_i \\
\nonumber
& \text{subject to}& \quad & 
\begin{aligned}[t]
zI+\sum_{i=1}^N\alpha_i\cdot\frac{1}{2}(e_i1^\top+1e_i^\top)-\sum_{i=1}^N\sum_{j=i}^N\beta_{ij}\cdot\frac{1}{2}(e_ie_j^\top+e_je_i^\top)+D&\succeq0\\
\beta&\geq0
\end{aligned}
\end{alignat}

\begin{theorem}[e.g., see~\cite{Mixon:15}]
Suppose the primal program \eqref{eq.primal} and dual program \eqref{eq.dual} are feasible and bounded.
\begin{itemize}
\item[(a)]
\textbf{Strong duality.}
The primal program \eqref{eq.primal} has optimal value $\operatorname{val}$ if and only if the dual program \eqref{eq.dual} has bounded optimal value $\operatorname{val}$.
\item[(b)]
\textbf{Complementary slackness.}
The decision variables $x$ and $y$ are optimal in \eqref{eq.primal} and \eqref{eq.dual}, respectively, if and only if
\[
\langle A^* y-c,x\rangle=0=\langle y,b-Ax\rangle.
\]
\end{itemize}
\end{theorem}

For notational simplicity, we organize indices according to clusters.
For example, from this point forward, $1_a$ denotes the indicator function of the $a$th cluster.
Also, we shuffle the rows and columns of $X$ and $D$ into blocks that correspond to clusters; for example, the $(i,j)$th entry of the $(a,b)$th block of $D$ is given by $D^{(a,b)}_{ij}$.
We also index $\alpha$ in terms of clusters; for example, the $i$th entry of the $a$th block of $\alpha$ is denoted $\alpha_{a,i}$.
For $\beta$, we identify
\[
\beta:=\sum_{i=1}^N\sum_{j=i}^N\beta_{ij}\cdot\frac{1}{2}(e_ie_j^\top+e_je_i^\top).
\]
Indeed, when $i\leq j$, the $(i,j)$th entry of $\beta$ is $\beta_{ij}$.
From this point forward, we consider $\beta$ as having its rows and columns shuffled according to clusters, so that the $(i,j)$th entry of the $(a,b)$th block is $\beta^{(a,b)}_{ij}$.

\begin{theorem}
\label{thm.integral optimality}
Take $X:=\sum_{a=1}^k\frac{1}{n_a}1_a1_a^\top$, where $n_a$ denotes the number of points in cluster $t$.
The following are equivalent:
\begin{itemize}
\item[(a)]
$X$ is a solution to the SDP relaxation \eqref{eq.kmeansSDP}.
\item[(b)]
Every solution to the dual SDP \eqref{eq.kmeansSDPdual} satisfies
\[
Q^{(a,a)}1=0,
\qquad
\beta^{(a,a)}=0
\qquad
\forall a\in\{1,\ldots,k\},
\]
where $Q:=A^*y-c$.
\item[(c)]
Every solution to the dual SDP \eqref{eq.kmeansSDPdual} satisfies
\[
\alpha_{a,r}
=-\frac{1}{n_a}z+\frac{1}{n_a^2}1^\top D^{(a,a)}1-\frac{2}{n_a}e_r^\top D^{(a,a)}1
\qquad
\forall a\in\{1,\ldots,k\},~r\in a.
\]
\end{itemize}
\end{theorem}

\begin{proof}
(a)$\Leftrightarrow$(b):
By complementary slackness, (a) is equivalent to having both
\begin{equation}
\label{eq.compslack1}
\langle A^* y-c,X\rangle=0
\end{equation}
and
\begin{equation}
\label{eq.compslack2}
\langle y,b-A(X)\rangle=0.
\end{equation}
Since $Q\succeq0$, we have
\[
\langle A^* y-c,X\rangle
=\langle Q,X\rangle
=\bigg\langle Q,\sum_{t=1}^k\frac{1}{n_t}1_t1_t^\top\bigg\rangle
=\sum_{t=1}^k\frac{1}{n_t}1_t^\top Q1_t
\geq0,
\]
with equality if and only if $Q1_a=0$ for every $a\in\{1,\ldots,k\}$.
Next, we recall that $y=z\oplus\alpha\oplus(-\beta)$, $b-A(X)\in L=0\oplus0\oplus\mathbb{R}_{\geq0}^{N(N+1)/2}$, and $b=k\oplus 1\oplus 0$.
As such, \eqref{eq.compslack2} is equivalent to $\beta$ having disjoint support with $\{\langle X,\frac{1}{2}(e_ie_j^\top+e_je_i^\top)\rangle\}_{i,j=1,i\leq j}^N$, i.e., $\beta^{(a,a)}=0$ for every cluster $a$.

(b)$\Rightarrow$(c):
Take any solution to the dual SDP \eqref{eq.kmeansSDPdual}, and note that
\begin{align*}
Q^{(a,a)}
&=zI+\bigg(\sum_{t=1}^k\sum_{i\in t}\alpha_{t,i}\cdot\frac{1}{2}(e_{t,i}1^\top+1e_{t,i}^\top)\bigg)^{(a,a)}-\beta^{(a,a)}+D^{(a,a)}\\
&=zI+\sum_{i\in a}\alpha_{a,i}\cdot\frac{1}{2}(e_i1^\top+1e_i^\top)+D^{(a,a)},
\end{align*}
where the $1$ vectors in the second line are $n_a$-dimensional (instead of $N$-dimensional, as in the first line), and similarly for $e_i$ (instead of $e_{t,i}$).
We now consider each entry of $Q^{(a,a)}1$, which is zero by assumption:
\begin{align}
0
\nonumber
&=e_r^\top Q^{(a,a)}1\\
\nonumber
&=e_r^\top \bigg(zI+\sum_{i\in a}\alpha_{a,i}\cdot\frac{1}{2}(e_i1^\top+1e_i^\top)+D^{(a,a)}\bigg)1\\
\nonumber
&=z+\sum_{i\in a}\alpha_{a,i}\cdot\frac{1}{2}(e_r^\top e_i1^\top1+e_r^\top 1e_i^\top1)+e_r^\top D^{(a,a)}1\\
\label{eq.linear system to solve}
&=z+\sum_{i\in a}\alpha_{a,i}\cdot\frac{1}{2}(n_a\delta_{ir}+1)+e_r^\top D^{(a,a)}1.
\end{align}
As one might expect, these $n_a$ linear equations determine the variables $\{\alpha_{a,i}\}_{i\in a}$.
To solve this system, we first observe
\begin{align*}
0
&=1^\top Q^{(a,a)}1\\
&=1^\top\bigg(zI+\sum_{i\in a}\alpha_{a,i}\cdot\frac{1}{2}(e_i1^\top+1e_i^\top)+D^{(a,a)}\bigg)1\\
&=n_az+\sum_{i\in a}\alpha_{a,i}\cdot\frac{1}{2}(1^\top e_i1^\top1+1^\top 1e_i^\top1)+1^\top D^{(a,a)}1\\
&=n_az+n_a\sum_{i\in a} \alpha_{a,i}+1^\top D^{(a,a)}1,
\end{align*}
and so rearranging gives
\[
\sum_{i\in a} \alpha_{a,i}
=-z-\frac{1}{n_a}1^\top D^{(a,a)}1.
\]
We use this identity to continue \eqref{eq.linear system to solve}:
\begin{align*}
0
&=z+\sum_{i\in a}\alpha_{a,i}\cdot\frac{1}{2}(n_a\delta_{ir}+1)+e_r^\top D^{(a,a)}1\\
&=z+\frac{n_a}{2}\alpha_{a,r}+\frac{1}{2}\sum_{i\in a}\alpha_{a,i}+e_r^\top D^{(a,a)}1\\
&=z+\frac{n_a}{2}\alpha_{a,r}+\frac{1}{2}\bigg(-z-\frac{1}{n_a}1^\top D^{(a,a)}1\bigg)+e_r^\top D^{(a,a)}1,
\end{align*}
and rearranging yields the desired formula for $\alpha_{a,r}$.

(c)$\Rightarrow$(a):
Take any solution to the dual SDP \eqref{eq.kmeansSDPdual}.
Then by assumption, the dual objective at this point is given by
\begin{align*}
kz+\sum_{t=1}^k\sum_{i\in t}\alpha_{t,i}
&=kz+\sum_{t=1}^k\sum_{i\in t}\bigg(-\frac{1}{n_t}z+\frac{1}{n_t^2}1^\top D^{(t,t)}1-\frac{2}{n_t}e_i^\top D^{(t,t)}1\bigg)\\
&=-\sum_{t=1}^k\frac{1}{n_t}1^\top D^{(t,t)}1\\
&=-\operatorname{Tr}(DX),
\end{align*}
i.e., the primal objective \eqref{eq.kmeansSDP} evaluated at $X$.
Since $X$ is feasible in the primal SDP, we conclude that $X$ is optimal by strong duality.
\end{proof}

\section{Finding a dual certificate} \label{sec:deterministic}

The goal is to certify when the SDP-optimal solution is integral.
In this event, Theorem~\ref{thm.integral optimality} characterizes acceptable dual certificates $(z,\alpha,\beta)$, but this information fails to uniquely determine a certificate.
In this section, we will motivate the application of additional constraints on dual certificates so as to identify certifiable instances.

We start by reviewing the characterization of dual certificates $(z,\alpha,\beta)$ provided in Theorem~\ref{thm.integral optimality}.
In particular, $\alpha$ is completely determined by $z$, and so $z$ and $\beta$ are the only remaining free variables.
Indeed, for every $a,b\in\{1,\ldots,k\}$, we have
\begin{align*}
&\bigg(\sum_{t=1}^k\sum_{i\in t}\alpha_{t,i}\cdot\frac{1}{2}(e_{t,i}1^\top+1e_{t,i}^\top)\bigg)^{(a,b)}\\
&\qquad=\sum_{i\in a}\alpha_{a,i}\cdot\frac{1}{2}e_i1^\top+\sum_{j\in b}\alpha_{b,j}\cdot\frac{1}{2}1e_j^\top\\
&\qquad=-\frac{1}{2}\bigg(\frac{1}{n_a}+\frac{1}{n_b}\bigg)z+\sum_{i\in a}\bigg(\frac{1}{n_a^2}1^\top D^{(a,a)}1-\frac{2}{n_a}e_i^\top D^{(a,a)}1\bigg)\frac{1}{2}e_i1^\top\\
&\qquad\qquad+\sum_{j\in b}\bigg(\frac{1}{n_b^2}1^\top D^{(b,b)}1-\frac{2}{n_b}e_j^\top D^{(b,b)}1\bigg)\frac{1}{2}1e_j^\top,
\end{align*}
and so since
\[
Q
=zI+\sum_{t=1}^k\sum_{i\in t}\alpha_{t,i}\cdot\frac{1}{2}(e_{t,i}1^\top+1e_{t,i}^\top)-\frac{1}{2}\beta+D,
\]
we may write $Q=z(I-E)+M-B$, where
\begin{align}
\label{eq.definition of E}
E^{(a,b)}
&:=\frac{1}{2}\bigg(\frac{1}{n_a}+\frac{1}{n_b}\bigg)11^\top\\
\nonumber
M^{(a,b)}
&:=D^{(a,b)}+\sum_{i\in a}\bigg(\frac{1}{n_a^2}1^\top D^{(a,a)}1-\frac{2}{n_a}e_i^\top D^{(a,a)}1\bigg)\frac{1}{2}e_i1^\top\\
\label{eq.definition of M}
&\qquad+\sum_{j\in b}\bigg(\frac{1}{n_b^2}1^\top D^{(b,b)}1-\frac{2}{n_b}e_j^\top D^{(b,b)}1\bigg)\frac{1}{2}1e_j^\top\\
\nonumber
B^{(a,b)}
&=\frac{1}{2}\beta^{(a,b)}
\end{align}
for every $a,b\in\{1,\ldots,k\}$.
The following is one way to formulate our task:
Given $D$ and a clustering (which in turn determines $E$ and $M$), determine whether there exist feasible $z$ and $B$ such that $Q\succeq0$; here, feasibility only requires $B$ to be symmetric with nonnegative entries and $B^{(a,a)}=0$ for every $a\in\{1,\ldots,k\}$.
We opt for a slightly more modest goal:
Find $z=z(D)$ and $B=B(D)$ such that $Q\succeq0$ for a large family of $D$'s.

Before determining $z$ and $B$, we first analyze $E$:

\begin{lemma}
\label{lemma.eigen of E}
Let $E$ be the matrix defined by \eqref{eq.definition of E}.
Then $\operatorname{rank}(E)\in\{1,2\}$.
The eigenvalue of largest magnitude is $\lambda\geq k$, and when $\operatorname{rank}(E)=2$, the other nonzero eigenvalue of $E$ is negative.
The eigenvectors corresponding to nonzero eigenvalues lie in the span of $\{1_a\}_{a=1}^k$.
\end{lemma}

\begin{proof}
Writing
\begin{equation*}
E
=\sum_{a=1}^k\sum_{b=1}^k\frac{1}{2}\bigg(\frac{1}{n_a}+\frac{1}{n_b}\bigg)1_a1_b^\top
=\frac{1}{2}\bigg(\sum_{a=1}^k\frac{1}{n_a}1_a\bigg)1^\top+\frac{1}{2}1\bigg(\sum_{b=1}^k\frac{1}{n_b}1_b\bigg)^\top,
\end{equation*}
we see that $\operatorname{rank}(E)\in\{1,2\}$, and it is easy to calculate $1^\top E1=Nk$ and $\operatorname{Tr}(E)=k$.
Observe that
\[
\lambda
=\sup_{\substack{x\in\mathbb{R}^N\\\|x\|=1}}x^\top Ex
\geq \frac{1}{N}1^\top E1
=k,
\]
and combining with $\operatorname{rank}(E)\leq2$ and $\operatorname{Tr}(E)=k$ then implies that the other nonzero eigenvalue (if there is one) is negative.
Finally, any eigenvector of $E$ with a nonzero eigenvalue necessarily lies in the column space of $E$, which is a subspace of $\operatorname{span}\{1_a\}_{a=1}^k$ by the definition of $E$.
\end{proof}

When finding $z$ and $B$ such that $Q=z(I-E)+M-B\succeq0$ it will be useful that $I-E$ has only one negative eigenvalue to correct.
Let $v$ denote the corresponding eigenvector.
Then we will pick $B$ so that $v$ is also an eigenvector of $M-B$.
Since we want $Q\succeq0$ for as many instances of $D$ as possible, we will then pick $z$ as large as possible, thereby sending $v$ to the nullspace of $Q$.
Unfortunately, the authors found that this constraint fails to uniquely determine $B$ in general.
Instead, we impose a stronger constraint:
\[
Q1_a=0
\qquad
\forall a\in\{1,\ldots,k\}.
\]
(This constraint implies $Qv=0$ by Lemma~\ref{lemma.eigen of E}.)
To see the implications of this constraint, note that we already necessarily have
\[
(Q1_a)_a
=\Big((z(I-E)+M-B)1_a\Big)_a
=z(I-E^{(a,a)})1+M^{(a,a)}1-B^{(a,a)}1
=z\bigg(1-\frac{1}{n_a}11^\top1\bigg)
=0,
\]
and so it remains to impose
\begin{align}
\nonumber
0
=(Q1_b)_a
&=\Big((z(I-E)+M-B)1_b\Big)_a\\
\label{eq.requirement for B}
&=-zE^{(a,b)}1+M^{(a,b)}1-B^{(a,b)}1
=-z\frac{n_a+n_b}{2n_a}1+M^{(a,b)}1-B^{(a,b)}1.
\end{align}
In order for there to exist a vector $B^{(a,b)}1\geq0$ that satisfies \eqref{eq.requirement for B}, $z$ must satisfy
\[
z\frac{n_a+n_b}{2n_a}\leq\min(M^{(a,b)}1),
\]
and since $z$ is independent of $(a,b)$, we conclude that
\begin{equation}
\label{eq.bound on z}
z\leq\min_{\substack{a,b\in\{1,\ldots,k\}\\a\neq b}}\frac{2n_a}{n_a+n_b}\min(M^{(a,b)}1).
\end{equation}
Again, in order to ensure $z(I-E)+M-B\succeq0$ for as many instances of $D$ as possible, we intend to choose $z$ as large as possible.
Luckily, there is a choice of $B$ which satisfies \eqref{eq.requirement for B} for every $(a,b)$, even when $z$ satisfies equality in \eqref{eq.bound on z}.
Indeed, we define
\begin{equation}
\label{eq.how to construct B}
u_{(a,b)}
:=M^{(a,b)}1-z\frac{n_a+n_b}{2n_a}1,
\qquad
\rho_{(a,b)}
:=u_{(a,b)}^\top 1,
\qquad
B^{(a,b)}
:=\frac{1}{\rho_{(b,a)}}u_{(a,b)}u_{(b,a)}^\top
\end{equation}
for every $a,b\in\{1,\ldots,k\}$ with $a\neq b$.
Then by design, $B$ immediately satisfies \eqref{eq.requirement for B}.
Also, note that $\rho_{(a,b)}=\rho_{(b,a)}$, and so $B^{(b,a)}=(B^{(a,b)})^\top$, meaning $B$ is symmetric.
Finally, we necessarily have $u_{(a,b)}\geq0$ (and thus $\rho_{(a,b)}\geq0$) by \eqref{eq.bound on z}, and we implicitly require $\rho_{(a,b)}>0$ for division to be permissible.
As such, we also have $B^{(a,b)}\geq0$, as desired.

Now that we have selected $z$ and $B$, it remains to check that $Q\succeq0$.
By construction, we already have $\Lambda:=\operatorname{span}\{1_a\}_{a=1}^k$ in the nullspace of $Q$, and so it suffices to ensure
\[
0
\preceq P_{\Lambda^\perp}QP_{\Lambda^\perp}
=P_{\Lambda^\perp}\Big(z(I-E)+M-B\Big)P_{\Lambda^\perp}
=zP_{\Lambda^\perp}+P_{\Lambda^\perp}(M-B)P_{\Lambda^\perp},
\]
which in turn is implied by
\[
\|P_{\Lambda^\perp}(M-B)P_{\Lambda^\perp}\|_{2\rightarrow2}
\leq z.
\]
To summarize, we have the following result:

\begin{theorem}
\label{thm.dual certificate}
Take $X:=\sum_{t=1}^k\frac{1}{n_t}1_t1_t^\top$, where $n_t$ denotes the number of points in cluster $t$.
Consider $M$ and $B$ defined by \eqref{eq.definition of M} and \eqref{eq.how to construct B}, respectively, and let $\Lambda$ denote the span of $\{1_t\}_{t=1}^k$.
Then $X$ is a solution to the SDP relaxation \eqref{eq.kmeansSDP} if
\begin{equation} \label{cert_condition}
\|P_{\Lambda^\perp}(M-B)P_{\Lambda^\perp}\|_{2\rightarrow2}
\leq \min_{\substack{a,b\in\{1,\ldots,k\}\\a\neq b}}\frac{2n_a}{n_a+n_b}\min(M^{(a,b)}1).
\end{equation}
\end{theorem}

A sufficient condition that implies Theorem \ref{thm.dual certificate} can be obtained by finding an upper bound on the left-hand side of \eqref{cert_condition}. This is Corollary \ref{cor.dual certificate}, which we use to prove the main theorem. 

\begin{corollary}
\label{cor.dual certificate}
Take $X:=\sum_{t=1}^k\frac{1}{n_t}1_t1_t^\top$, where $n_t$ denotes the number of points in cluster $t$.
Let $\Psi$ denote the $m\times N$ matrix whose $(a,i)$th column is $x_{a,i}-c_a$, where
\[
c_a:=\frac{1}{n_a}\sum_{i\in a}x_{a,i}
\]
denotes the empirical center of cluster $a$.
Consider $M$ and $\rho_{(a,b)}$ defined by \eqref{eq.definition of M} and \eqref{eq.how to construct B}, respectively.
Then $X$ is a solution to the SDP relaxation \eqref{eq.kmeansSDP} if
\[
2\|\Psi\|_{2\rightarrow2}^2+\sum_{a=1}^k\sum_{b=a+1}^k\frac{\|P_{1^\perp}M^{(a,b)}1\|_2\|P_{1^\perp}M^{(b,a)}1\|_2}{\rho_{(a,b)}}
\leq\min_{\substack{a,b\in\{1,\ldots,k\}\\a\neq b}}\frac{2n_a}{n_a+n_b}\min(M^{(a,b)}1).
\]
\end{corollary}

\begin{proof}
First, the triangle inequality gives
\begin{equation}
\label{eq.triangle for new condition}
\|P_{\Lambda^\perp}(M-B)P_{\Lambda^\perp}\|_{2\rightarrow2}
\leq\|P_{\Lambda^\perp}MP_{\Lambda^\perp}\|_{2\rightarrow2}+\|P_{\Lambda^\perp}BP_{\Lambda^\perp}\|_{2\rightarrow2}.
\end{equation}
We will bound the terms in \eqref{eq.triangle for new condition} separately and then combine the bounds to derive a sufficient condition for Theorem~\ref{thm.dual certificate}.
To bound the first term in \eqref{eq.triangle for new condition}, let $\nu$ be the $N\times 1$ vector whose $(a,i)$th entry is $\|x_{a,i}\|^2$, and let $\Phi$ be the $m\times N$ matrix whose $(a,i)$th column is $x_{a,i}$.
Then
\[
D_{(a,i),(b,j)}
=\|x_{a,i}-x_{b,j}\|^2
=\|x_{a,i}\|^2-2x_{a,i}^\top x_{b,j}+\|x_{b,j}\|^2
=(\nu1^\top-2\Phi^\top\Phi+1\nu^\top)_{(a,i),(b,j)},
\]
meaning $D=\nu1^\top-2\Phi^\top\Phi+1\nu^\top$.
With this, we appeal to the blockwise definition of $M$ \eqref{eq.definition of M}:
\begin{align*}
\|P_{\Lambda^\perp}MP_{\Lambda^\perp}\|_{2\rightarrow2}
=\|P_{\Lambda^\perp}DP_{\Lambda^\perp}\|_{2\rightarrow2}
&=\|P_{\Lambda^\perp}(\nu1^\top-2\Phi^\top\Phi+1\nu^\top)P_{\Lambda^\perp}\|_{2\rightarrow2}\\
&=2\|P_{\Lambda^\perp}\Phi^\top\Phi P_{\Lambda^\perp}\|_{2\rightarrow2}
=2\|\Phi P_{\Lambda^\perp}\|_{2\rightarrow2}^2
=2\|\Psi\|_{2\rightarrow2}^2.
\end{align*}
For the second term in \eqref{eq.triangle for new condition}, we first write the decomposition
\[
B=\sum_{a=1}^k\sum_{b=a+1}^k\Big(H_{(a,b)}(B^{(a,b)})+H_{(b,a)}(B^{(b,a)})\Big),
\]
where $H_{(a,b)}\colon\mathbb{R}^{n_a\times n_b}\rightarrow\mathbb{R}^{N\times N}$ produces a matrix whose $(a,b)$th block is the input matrix, and is otherwise zero.
Then
\begin{align*}
P_{\Lambda^\perp}BP_{\Lambda^\perp}
&=\sum_{a=1}^k\sum_{b=a+1}^kP_{\Lambda^\perp}\Big(H_{(a,b)}(B^{(a,b)})+H_{(b,a)}(B^{(b,a)})\Big)P_{\Lambda^\perp}\\
&=\sum_{a=1}^k\sum_{b=a+1}^k\Big(H_{(a,b)}(P_{1^\perp}B^{(a,b)}P_{1^\perp})+H_{(b,a)}(P_{1^\perp}B^{(b,a)}P_{1^\perp})\Big),
\end{align*}
and so the triangle inequality gives
\begin{align*}
\|P_{\Lambda^\perp}BP_{\Lambda^\perp}\|_{2\rightarrow2}
&\leq\sum_{a=1}^k\sum_{b=a+1}^k\|H_{(a,b)}(P_{1^\perp}B^{(a,b)}P_{1^\perp})+H_{(b,a)}(P_{1^\perp}B^{(b,a)}P_{1^\perp})\|_{2\rightarrow2}\\
&=\sum_{a=1}^k\sum_{b=a+1}^k\|P_{1^\perp}B^{(a,b)}P_{1^\perp}\|_{2\rightarrow2},
\end{align*}
where the last equality can be verified by considering the spectrum of the square:
\begin{align*}
&\Big(H_{(a,b)}(P_{1^\perp}B^{(a,b)}P_{1^\perp})+H_{(b,a)}(P_{1^\perp}B^{(b,a)}P_{1^\perp})\Big)^2\\
&\qquad=H_{(a,a)}\Big((P_{1^\perp}B^{(a,b)}P_{1^\perp})(P_{1^\perp}B^{(a,b)}P_{1^\perp})^\top\Big)+H_{(b,b)}\Big((P_{1^\perp}B^{(a,b)}P_{1^\perp})^\top(P_{1^\perp}B^{(a,b)}P_{1^\perp})\Big).
\end{align*}
At this point, we use the definition of $B$ \eqref{eq.how to construct B} to get
\[
\|P_{1^\perp}B^{(a,b)}P_{1^\perp}\|_{2\rightarrow2}
=\frac{\|P_{1^\perp}u_{(a,b)}\|_2\|P_{1^\perp}u_{(b,a)}\|_2}{\rho_{(a,b)}}.
\]
Recalling the definition of $u_{(a,b)}$ \eqref{eq.how to construct B} and combining these estimates then produces the result.
\end{proof}

\section{Proof of main result} \label{sec:model}

In this section, we apply the certificate from Corollary \ref{cor.dual certificate} to the  $(\mathcal{D},\gamma,n)$-stochastic ball model (see Definition \ref{stochastic_balls}) to prove our main result. We will prove Theorem \ref{main_theorem} with the help of several lemmas.

\begin{lemma}
Denote
\[
c_a:=\frac{1}{n}\sum_{i=1}^nx_{a,i},
\qquad
\Delta_{ab}:=\|\gamma_a-\gamma_b\|,
\qquad
O_{ab}:=\frac{\gamma_a+\gamma_b}{2}.
\]
Then the $(\mathcal{D},\gamma,n)$-stochastic ball model satisfies the following estimates:
\begin{align}
\label{eq.empirical center}
\|c_a-\gamma_a\|&<\epsilon\qquad\mbox{w.p.}\qquad1-e^{-\Omega_{m,\epsilon}(n)}\\
\label{eq.empirical average radius}
\bigg|\frac{1}{n}\sum_{i=1}^n\|r_{a,i}\|^2-\mathbb{E}\|r\|^2\bigg|&<\epsilon\qquad\mbox{w.p.}\qquad1-e^{-\Omega_\epsilon(n)}\\
\label{eq.empirical average distance from midpoint}
\bigg|\frac{1}{n}\sum_{i=1}^n\|x_{a,i}-O_{ab}\|^2-\mathbb{E}\|r+\gamma_a-O_{ab}\|^2\bigg|&<\epsilon\qquad\mbox{w.p.}\qquad1-e^{-\Omega_{\Delta_{ab},\epsilon}(n)}
\end{align}
\end{lemma}

\begin{proof}
Since $\mathbb{E}r=0$ and $\|r\|^2\leq1$ almost surely, one may lift
\[
X_{a,i}:=\left[\begin{array}{cc}0&r_{a,i}^\top\\r_{a,i}&0\end{array}\right]
\]
and apply the Matrix Hoeffding inequality~\cite{Tropp:10} to conclude that
\[
\operatorname{Pr}\bigg(\bigg\|\sum_{i=1}^nr_{a,i}\bigg\|\geq t\bigg)\leq me^{-t^2/8n}.
\]
Taking $t:=\epsilon n$ then gives \eqref{eq.empirical center}.
For \eqref{eq.empirical average radius} and \eqref{eq.empirical average distance from midpoint}, notice that the random variables in each sum are iid and confined to an interval almost surely, and so the result follows from Hoeffding's inequality.
\end{proof}

\begin{lemma}
\label{lemma.difference of distances}
Under the $(\mathcal{D},\gamma,n)$-stochastic ball model, we have $D^{(a,b)}1-D^{(a,a)}1=4np+q$, where
\begin{align*}
	p_i	&:=r_{a,i}^\top(\gamma_a-O_{ab})+\frac{\Delta_{ab}^2}{4}\\
	q_i	&:=2n(x_{a,i}-O_{ab})^\top\bigg((c_a-c_b)-(\gamma_a-\gamma_b)\bigg)+\bigg(\sum_{j=1}^n\|x_{b,j}-O_{ab}\|^2-\sum_{j=1}^n\|x_{a,j}-O_{ab}\|^2\bigg)
\end{align*}
and $|q_i|\leq(6+\Delta_{ab})n\epsilon$ with probability $1-e^{-\Omega_{m,\Delta_{ab},\epsilon}(n)}$.
\end{lemma}

\begin{proof}
Add and subtract $O_{ab}$ and then expand the squares to get
\begin{align*}
e_i^\top(D^{(a,b)}1-D^{(a,a)}1)
&=\sum_{j=1}^n\|x_{a,i}-x_{b,j}\|^2-\sum_{j=1}^n\|x_{a,i}-x_{a,j}\|^2\\
&=n\bigg(-2(x_{a,i}-O_{ab})^\top(c_b-O_{ab})+\frac{1}{n}\sum_{j=1}^n\|x_{b,j}-O_{ab}\|^2\bigg)\\
&\qquad-n\bigg(-2(x_{a,i}-O_{ab})^\top(c_a-O_{ab})+\frac{1}{n}\sum_{j=1}^n\|x_{a,j}-O_{ab}\|^2\bigg)\\
&=2n(x_{a,i}-O_{ab})^\top(c_a-c_b)+\bigg(\sum_{j=1}^n\|x_{b,j}-O_{ab}\|^2-\sum_{j=1}^n\|x_{a,j}-O_{ab}\|^2\bigg).
\end{align*}
Add and subtract $\gamma_a-\gamma_b$ to $c_a-c_b$ and distribute over the resulting sum to obtain
\begin{align*}
e_i^\top(D^{(a,b)}1-D^{(a,a)}1)
&=2n(x_{a,i}-O_{ab})^\top(\gamma_a-\gamma_b)+q\\
&=4n\Big(r_{a,i}+(\gamma_a-O_{ab})\Big)^\top(\gamma_a-O_{ab})+q.
\end{align*}
Distributing and identifying $\|\gamma_a-O_{ab}\|^2=\Delta_{ab}^2/4$ explains the definition of $p$.
To show $|q_i|\leq(6+\Delta_{ab})n\epsilon$, apply triangle and Cauchy-Schwarz inequalities to obtain
\begin{align*}
|q_i|
	&\leq \bigg|2n(x_{a,i}-O_{ab})^\top\bigg((c_a-c_b)-(\gamma_a-\gamma_b)\bigg)\bigg|+\bigg|\sum_{j=1}^n\|x_{b,j}-O_{ab}\|^2-\sum_{j=1}^n\|x_{a,j}-O_{ab}\|^2\bigg|\\
	&\leq 2n \bigg(\|r_{a,i}\|+\|\gamma_a-O_{a,b}\|\bigg)\bigg(\|c_a-\gamma_a\|+\|c_b-\gamma_b\|\bigg)+\bigg|\sum_{j=1}^n\|x_{b,j}-O_{ab}\|^2-\sum_{j=1}^n\|x_{a,j}-O_{ab}\|^2\bigg|\\
	&\leq 2n\bigg(1+\frac{\Delta_{ab}}{2}\bigg)\bigg(\|c_a-\gamma_a\|+\|c_b-\gamma_b\|\bigg)+\bigg|\sum_{j=1}^n\|x_{b,j}-O_{ab}\|^2-\sum_{j=1}^n\|x_{a,j}-O_{ab}\|^2\bigg|.
\end{align*}
To finish the argument, apply \eqref{eq.empirical center} to the first term while adding and subtracting
\[
\mathbb{E}\|r+\gamma_a-O_{ab}\|^2=\mathbb{E}\|r+\gamma_b-O_{ab}\|^2,
\]
from the second and apply \eqref{eq.empirical average distance from midpoint}.
\end{proof}

\begin{lemma}
\label{lemma.within cluster distances}
Under the $(\mathcal{D},\gamma,n)$-stochastic ball model, we have
\[
\bigg|\frac{1}{n}1^\top D^{(a,a)}1-2n\mathbb{E}\|r\|^2\bigg|
\leq 4n\epsilon
\qquad
\mbox{w.p.}
\qquad
1-e^{-\Omega_{\Delta_{ab},\epsilon}(n)}.
\]
\end{lemma}

\begin{proof}
Add and subtract $\gamma_a$ and expand the square to get
\[
\frac{1}{n}e_i^\top D^{(a,a)}1
=\frac{1}{n}\sum_{j=1}^n\|x_{a,i}-x_{a,j}\|^2
=\|r_{a,i}\|^2-2r_{a,i}^\top(c_a-\gamma_a)+\frac{1}{n}\sum_{j=1}^n\|r_{a,j}\|^2.
\]
The triangle and Cauchy--Schwarz inequalities then give
\begin{align*}
&\bigg|\frac{1}{n}1^\top D^{(a,a)}1-2n\mathbb{E}\|r\|^2\bigg|\\
&\qquad=\bigg|\sum_{i=1}^n\bigg(\|r_{a,i}\|^2-2r_{a,i}^\top(c_a-\gamma_a)+\frac{1}{n}\sum_{j=1}^n\|r_{a,j}\|^2\bigg)-2n\mathbb{E}\|r\|^2\bigg|\\
&\qquad\leq n\bigg|\frac{1}{n}\sum_{i=1}^n\|r_{a,i}\|^2-\mathbb{E}\|r\|^2\bigg|+2\sum_{i=1}^n|r_{a,i}^\top(c_a-\gamma_a)|+n\bigg|\frac{1}{n}\sum_{j=1}^n\|r_{a,j}\|^2-\mathbb{E}\|r\|^2\bigg|\\
&\qquad\leq n\bigg|\frac{1}{n}\sum_{i=1}^n\|r_{a,i}\|^2-\mathbb{E}\|r\|^2\bigg|+2\sum_{i=1}^n\|c_a-\gamma_a\|+n\bigg|\frac{1}{n}\sum_{j=1}^n\|r_{a,j}\|^2-\mathbb{E}\|r\|^2\bigg|\\
&\qquad\leq 4n\epsilon,
\end{align*}
where the last step occurs with probability $1-e^{-\Omega_{\Delta_{ab},\epsilon}(n)}$ by a union bound over \eqref{eq.empirical average radius} and \eqref{eq.empirical center}.
\end{proof}

\begin{lemma}
\label{lemma.difference of average distances}
Under the $(\mathcal{D},\gamma,n)$-stochastic ball model, we have
\[
1^\top D^{(a,b)}1-1^\top D^{(a,a)}1
\geq n^2\Delta_{ab}^2-(6+3\Delta_{ab})n^2\epsilon
\qquad
\mbox{w.p.}
\qquad
1-e^{-\Omega_{m,\Delta_{ab},\epsilon}(n)}.
\]
\end{lemma}

\begin{proof}
Lemma~\ref{lemma.difference of distances} gives
\begin{align*}
1^\top D^{(a,b)}1-1^\top D^{(a,a)}1
&=1^\top(4np+q)\\
&\geq 4n\sum_{i=1}^n\bigg(r_{a,i}^\top(\gamma_a-O_{ab})+\|\gamma_a-O_{ab}\|^2\bigg)-(6+\Delta_{ab})n^2\epsilon\\
&\geq 4n\bigg(n(c_a-\gamma_a)^\top(\gamma_a-O_{ab})+\frac{n\Delta_{ab}^2}{4}\bigg)-(6+\Delta_{ab})n^2\epsilon.
\end{align*}
Cauchy--Schwarz along with \eqref{eq.empirical center} then gives the result.
\end{proof}

\begin{lemma}
\label{lemma.bound on rhs}
Under the $(\mathcal{D},\gamma,n)$-stochastic ball model, there exists $C=C(\gamma)$ such that
\[
\min_{\substack{a,b\in\{1,\ldots,k\}\\a\neq b}}\min(M^{(a,b)}1)
\geq n\Delta(\Delta-2)+Cn\epsilon
\qquad
\mbox{w.p.}
\qquad
1-e^{-\Omega_{m,\gamma,\epsilon}(n)},
\]
where $\displaystyle{\Delta:=\min_{\substack{a,b\in\{1,\ldots,k\}\\a\neq b}}\Delta_{ab}}$.
\end{lemma}

\begin{proof}
Fix $a$ and $b$.
Then by Lemma~\ref{lemma.difference of distances}, the following holds with probability $1-e^{-\Omega_{m,\Delta_{ab},\epsilon}(n)}$:
\begin{align*}
\min\Big(D^{(a,b)}1-D^{(a,a)}1\Big)
&\geq 4n\min_{i\in\{1,\ldots,n\}}\bigg(r_{a,i}^\top(\gamma_a-O_{ab})+\frac{\Delta_{ab}^2}{4}\bigg)-(6+\Delta_{ab})n\epsilon\\
&\geq n\Delta_{ab}^2-2n\Delta_{ab}-(6+\Delta_{ab})n\epsilon,
\end{align*}
where the last step is by Cauchy--Schwarz.
Taking a union bound with Lemma~\ref{lemma.within cluster distances} then gives
\begin{align*}
&\min(M^{(a,b)}1)\\
&\qquad=\min\Big(D^{(a,b)}1-D^{(a,a)}1\Big)+\frac{1}{2}\bigg(\frac{1}{n}1^\top D^{(a,a)}1-\frac{1}{n}1^\top D^{(b,b)}1\bigg)\\
&\qquad\geq\min\Big(D^{(a,b)}1-D^{(a,a)}1\Big)
-\frac{1}{2}\bigg(\bigg|\frac{1}{n}1^\top D^{(a,a)}1-2n\mathbb{E}\|r\|^2\bigg|+\bigg|\frac{1}{n}1^\top D^{(b,b)}1-2n\mathbb{E}\|r\|^2\bigg|\bigg)\\
&\qquad\geq n\Delta_{ab}(\Delta_{ab}-2)-(10+\Delta_{ab})n\epsilon
\end{align*}
with probability $1-e^{-\Omega_{\Delta_{ab},\epsilon}(n)}$.
The result then follows from a union bound over $a$ and $b$.
\end{proof}

\begin{lemma}
\label{lemma.bound on numerator}
Suppose $\epsilon\leq 1$.
Then there exists $C=C(\Delta_{ab},m)$ such that under the $(\mathcal{D},\gamma,n)$-stochastic ball model, we have
\[
\|P_{1^\perp}M^{(a,b)}1\|^2
\leq \frac{4n^3\Delta_{ab}^2}{m}+Cn^3\epsilon
\]
with probability $1-e^{-\Omega_{m,\Delta_{ab},\epsilon}(n)}$.
\end{lemma}

\begin{proof}
First, a quick calculation reveals
\begin{align*}
e_i^\top M^{(a,b)}1
&=e_i^\top D^{(a,b)}1-e_i^\top D^{(a,a)}1+\frac{1}{2}\bigg(\frac{1}{n}1^\top D^{(a,a)}1-\frac{1}{n}1^\top D^{(b,b)}1\bigg),\\
\frac{1}{n}1^\top M^{(a,b)}1
&=\frac{1}{n}1^\top D^{(a,b)}1-\frac{1}{2}\bigg(\frac{1}{n}1^\top D^{(a,a)}1+\frac{1}{n}1^\top D^{(b,b)}1\bigg),
\end{align*}
from which it follows that
\begin{align*}
e_i^\top P_{1^\perp}M^{(a,b)}1
&=e_i^\top M^{(a,b)}1-\frac{1}{n}1^\top M^{(a,b)}1\\
&=\bigg(e_i^\top D^{(a,b)}1-\frac{1}{n}1^\top D^{(a,b)}1\bigg)-\bigg(e_i^\top D^{(a,a)}1-\frac{1}{n}1^\top D^{(a,a)}1\bigg)\\
&=e_i^\top P_{1^\perp}(D^{(a,b)}1-D^{(a,a)}1).
\end{align*}
As such, we have
\begin{align}
\nonumber
\|P_{1^\perp}M^{(a,b)}1\|^2
&=\|P_{1^\perp}(D^{(a,b)}1-D^{(a,a)}1)\|^2\\
\label{eq.two terms to bound}
&=\|D^{(a,b)}1-D^{(a,a)}1\|^2-\|P_1(D^{(a,b)}1-D^{(a,a)}1)\|^2.
\end{align}
To bound the first term, we apply the triangle inequality over Lemma~\ref{lemma.difference of distances}:
\begin{equation}
\label{eq.bound of first term 1}
\|D^{(a,b)}1-D^{(a,a)}1\|
\leq 4n\|p\|+\|q\|
\leq 4n\|p\|+(6+\Delta_{ab})n^{3/2}\epsilon.
\end{equation}
We proceed by bounding $\|p\|$.
To this end, note that the $p_i$'s are iid random variables whose outcomes lie in a finite interval (of width determined by $\Delta_{ab}$) with probability $1$.
As such, Hoeffding's inequality gives
\[
\bigg|\frac{1}{n}\sum_{i=1}^n p_i^2-\mathbb{E}p_1^2\bigg|
\leq \epsilon
\qquad
\mbox{w.p.}
\qquad
1-e^{-\Omega_{\Delta_{ab},\epsilon}(n)}.
\]
With this, we then have
\begin{equation}
\label{eq.bound of first term 2}
\|p\|^2
=n\bigg(\frac{1}{n}\sum_{i=1}^np_i^2-\mathbb{E}p_1^2+\mathbb{E}p_1^2\bigg)
\leq n\mathbb{E}p_1^2+n\epsilon
\end{equation}
in the same event.
To determine $\mathbb{E}p_1^2$, first take $r_1:=e_1^\top r$.
Then since the distribution of $r$ is rotation invariant, we may write
\[
p_1
=r_{a,1}^\top(\gamma_a-O_{ab})+\|\gamma_a-O_{ab}\|^2
=\frac{\Delta_{ab}}{2}r_1+\frac{\Delta_{ab}^2}{4},
\]
where the second equality above is equality in distribution.
We then have
\begin{equation}
\label{eq.bound of first term 3}
\mathbb{E}p_1^2
=\mathbb{E}\bigg(\frac{\Delta_{ab}}{2}r_1+\frac{\Delta_{ab}^2}{4}\bigg)^2
=\frac{\Delta_{ab}^2}{4}\mathbb{E}r_1^2+\frac{\Delta_{ab}^4}{16}.
\end{equation}
We also note that $1\geq\mathbb{E}\|r\|^2=m\mathbb{E}r_1^2$ by linearity of expectation, and so
\begin{equation}
\label{eq.bound of first term 4}
\mathbb{E}r_1^2\leq\frac{1}{m}.
\end{equation}
Combining \eqref{eq.bound of first term 1}, \eqref{eq.bound of first term 2}, \eqref{eq.bound of first term 3} and \eqref{eq.bound of first term 4} then gives
\begin{equation}
\label{eq.bound of first term summary}
\|D^{(a,b)}1-D^{(a,a)}1\|
\leq\bigg(\frac{4n^3\Delta_{ab}^2}{m}+n^3\Delta_{ab}^4+16n^3\epsilon\bigg)^{1/2}+(6+\Delta_{ab})n^{3/2}\epsilon.
\end{equation}
To bound the second term of \eqref{eq.two terms to bound}, first note that
\begin{equation}
\label{eq.second term 1}
\|P_1(D^{(a,b)}1-D^{(a,a)}1)\|
=\frac{1}{\sqrt{n}}\Big|1^\top D^{(a,b)}1-1^\top D^{(a,a)}1\Big|.
\end{equation}
Lemma~\ref{lemma.difference of average distances} then gives
\begin{equation}
\label{eq.second term 2}
\Big|1^\top D^{(a,b)}1-1^\top D^{(a,a)}1\Big|
\geq 1^\top D^{(a,b)}1-1^\top D^{(a,a)}1
\geq n^2\Delta_{ab}^2-(6+3\Delta_{ab})n^2\epsilon
\end{equation}
with probability $1-e^{-\Omega_{m,\Delta_{ab},\epsilon}(n)}$.
Using \eqref{eq.two terms to bound} to combine \eqref{eq.bound of first term summary} with \eqref{eq.second term 1} and \eqref{eq.second term 2} then gives the result.
\end{proof}

\begin{lemma}
\label{lemma.bound on rho}
There exists $C=C(\gamma)$ such that under the $(\mathcal{D},\gamma,n)$-stochastic ball model, we have
\[
\rho_{(a,b)}
\geq
2n^2\Delta_{ab}^2-Cn^2\epsilon
\qquad
\mbox{w.p.}
\qquad
1-e^{-\Omega_{\mathcal{D},\gamma,\epsilon}(n)}.
\]
\end{lemma}

\begin{proof}
Recall from \eqref{eq.how to construct B} that
\begin{equation}
\label{eq.decomposition of rho}
\rho_{(a,b)}
=u_{(a,b)}^\top 1
=1^\top M^{(a,b)}1-nz
=1^\top M^{(a,b)}1-n\min_{\substack{a,b\in\{1,\ldots,k\}\\a\neq b}}\min(M^{(a,b)}1).
\end{equation}
To bound the first term, we leverage Lemma~\ref{lemma.difference of average distances}:
\begin{align*}
1^\top M^{(a,b)}1
&=1^\top D^{(a,b)}1-\frac{1}{2}(1^\top D^{(a,a)}1+1^\top D^{(b,b)}1)\\
&=\frac{1}{2}\Big(1^\top D^{(a,b)}1-1^\top D^{(a,a)}1\Big)+\frac{1}{2}\Big(1^\top D^{(b,a)}1-1^\top D^{(b,b)}1\Big)\\
&\geq n^2\Delta_{ab}^2-(6+3\Delta_{ab})n^2\epsilon
\end{align*}
with probability $1-e^{-\Omega_{m,\Delta_{ab},\epsilon}(n)}$.
To bound the second term in \eqref{eq.decomposition of rho}, note from Lemma~\ref{lemma.within cluster distances} that
\begin{align*}
&\min(M^{(a,b)}1)\\
&\qquad=\min\Big(D^{(a,b)}1-D^{(a,a)}1\Big)+\frac{1}{2}\bigg(\frac{1}{n}1^\top D^{(a,a)}1-\frac{1}{n}1^\top D^{(b,b)}1\bigg)\\
&\qquad\leq\min\Big(D^{(a,b)}1-D^{(a,a)}1\Big)
+\frac{1}{2}\bigg(\bigg|\frac{1}{n}1^\top D^{(a,a)}1-2n\mathbb{E}\|r\|^2\bigg|+\bigg|\frac{1}{n}1^\top D^{(b,b)}1-2n\mathbb{E}\|r\|^2\bigg|\bigg)\\
&\qquad\leq\min\Big(D^{(a,b)}1-D^{(a,a)}1\Big)+4n\epsilon
\end{align*}
with probability $1-e^{-\Omega_{\Delta_{ab},\epsilon}(n)}$.
Next, Lemma~\ref{lemma.difference of distances} gives
\[
\min\Big(D^{(a,b)}1-D^{(a,a)}1\Big)
\leq n\Delta_{ab}^2+(6+\Delta_{ab})n\epsilon+4n\min_{i\in\{1,\ldots,n\}}r_{a,i}^\top(\gamma_a-O_{ab}).
\]
By assumption, we know $\|r\|\geq1-\epsilon$ with positive probability regardless of $\epsilon>0$.
It then follows that
\[
r^\top(\gamma_a-O_{ab})
\leq-\frac{\Delta_{ab}}{2}+\epsilon
\]
with some ($\epsilon$-dependent) positive probability.
As such, we may conclude that 
\[
\min_{i\in\{1,\ldots,n\}}r_{a,i}^\top(\gamma_a-O_{ab})
\leq-\frac{\Delta_{ab}}{2}+\epsilon
\qquad
\mbox{w.p.}
\qquad
1-e^{-\Omega_{\mathcal{D},\epsilon}(n)}.
\]
Combining these estimates then gives
\[
\min(M^{(a,b)}1)
\leq n\Delta_{ab}^2-2n\Delta_{ab}+(14+\Delta_{ab})n\epsilon
\qquad
\mbox{w.p.}
\qquad
1-e^{-\Omega_{\mathcal{D},\Delta_{ab},\epsilon}(n)}.
\]
Performing a union bound over $a$ and $b$ then gives
\[
\min_{\substack{a,b\in\{1,\ldots,k\}\\a\neq b}}\min(M^{(a,b)}1)
\leq  n\Delta^2-2n\Delta+(14+\Delta)n\epsilon
\qquad
\mbox{w.p.}
\qquad
1-e^{-\Omega_{\mathcal{D},\gamma,\epsilon}(n)}.
\]
Combining these estimates then gives the result.
\end{proof}

\begin{lemma}
\label{lemma.bound on spectral norm of psi}
Under the $(\mathcal{D},\gamma,n)$-stochastic ball model, we have
\[
\|\Psi\|_{2\rightarrow2}
\leq\bigg(\frac{(1+\epsilon)\sigma}{\sqrt{m}}+\epsilon\bigg)\sqrt{N}
\qquad
\mbox{w.p.}
\qquad
1-e^{-\Omega_{m,k,\sigma,\epsilon}(n)},
\]
where $\sigma^2:=\mathbb{E}\|r\|^2$ for $r\sim\mathcal{D}$.
\end{lemma}

\begin{proof}
Let $R$ denote the matrix whose $(a,i)$th column is $r_{a,i}$.
Then
\[
\Psi
=R-\Big[(c_1-\gamma_1)1^\top~\cdots~(c_k-\gamma_k)1^\top\Big],
\]
and so the triangle inequality gives
\[
\|\Psi\|_{2\rightarrow2}
\leq\|R\|_{2\rightarrow2}+\Big\|\Big[(c_1-\gamma_1)1^\top~\cdots~(c_k-\gamma_k)1^\top\Big]\Big\|_{2\rightarrow2}
\leq\|R\|_{2\rightarrow2}+\bigg(n\sum_{a=1}^k\|c_a-\gamma_a\|^2\bigg)^{1/2},
\]
where the last estimate passes to the Frobenius norm.
For the first term, since $\mathcal{D}$ is rotation invariant, we may apply Theorem~5.41 in~\cite{vershynin11}:
\[
\|R\|_{2\rightarrow2}
\leq(1+\epsilon)\sigma\sqrt{\frac{N}{m}}
\qquad
\mbox{w.p.}
\qquad
1-e^{-\Omega_{m,\sigma,\epsilon}(n)}.
\]
For the second term, apply \eqref{eq.empirical center}.
The union bound then gives the result.
\end{proof}

\begin{proof}[Proof of Theorem~\ref{main_theorem}]
First, we combine Lemmas~\ref{lemma.bound on numerator}, \ref{lemma.bound on rho} and \ref{lemma.bound on spectral norm of psi}:
For each $\epsilon>0$, we have
\[
2\|\Psi\|_{2\rightarrow2}^2+\sum_{a=1}^k\sum_{b=a+1}^k\frac{\|P_{1^\perp}M^{(a,b)}1\|_2\|P_{1^\perp}M^{(b,a)}1\|_2}{\rho_{(a,b)}}
\leq 2\bigg(\frac{1+\epsilon}{\sqrt{m}}+\epsilon\bigg)^2nk+\sum_{a=1}^k\sum_{b=a+1}^k\frac{4n^3\Delta_{ab}^2/m+Cn^3\epsilon}{2n^2\Delta-Cn^2\epsilon}
\]
with probability $1-e^{-\Omega_{\mathcal{D},\gamma,\epsilon}(n)}$.
Furthermore, for every $\delta>0$, there exists an $\epsilon>0$ such that 
\[
2\bigg(\frac{1+\epsilon}{\sqrt{m}}+\epsilon\bigg)^2nk+\sum_{a=1}^k\sum_{b=a+1}^k\frac{4n^3\Delta_{ab}^2/m+Cn^3\epsilon}{2n^2\Delta-Cn^2\epsilon}
\leq n\bigg[\bigg(\frac{2k}{m}+\frac{k(k-1)\Delta\operatorname{Cond}(\gamma)}{m}\bigg)+\delta\bigg].
\]
Considering Lemma~\ref{lemma.bound on rhs}, it then suffices to have
\[
\frac{2k}{m}+\frac{k(k-1)\Delta\operatorname{Cond}(\gamma)}{m}
< \Delta(\Delta-2).
\]
Rearranging then gives
\[
\Delta
>2+\frac{2k}{m\Delta}+\frac{k(k-1)\operatorname{Cond}(\gamma)}{m},
\]
which is implied by the hypothesis since $\Delta\geq2$ and $\operatorname{Cond}(\gamma)\geq1$.
\end{proof}

\section*{Acknowledgments}

DGM was supported by NSF Grant No.\ DMS-1321779.
The views expressed in this article are those of the authors and do not reflect the official policy or position
of the United States Air Force, Department of Defense, or the U.S.\ Government.

\bibliographystyle{abbrv}
\bibliography{clustering}

\end{document}